\documentclass{article}

\usepackage{style}
\usepackage[margin=1in]{geometry}

\begin{document}

\title{Non-adaptive Learning of Random Hypergraphs with Queries}
\author{Bethany Austhof \thanks{University of Illinois at Chicago. Email: \texttt{bausth2@uic.edu}}
\and
Lev Reyzin \thanks{University of Illinois at Chicago. Email: \texttt{lreyzin@uic.edu}}
\and
Erasmo Tani\thanks{Sapienza University of Rome. Email: \texttt{tani@di.uniroma1.it}}}

\maketitle

\begin{abstract}
We study the problem of learning a hidden hypergraph $G=(V,E)$ by making a single batch of queries (non-adaptively). We consider the hyperedge detection model, in which every query must be of the form:
``{\em Does this set $S\subseteq V$ contain at least one full hyperedge?}'' In this model, it is known \citep{abasi2019learning} that there is no algorithm that allows to non-adaptively learn arbitrary hypergraphs by making fewer than $\Omega(\min\{m^2\log n, n^2\})$ even when the hypergraph is constrained to be $2$-uniform (i.e. the hypergraph is simply a graph). Recently, \cite{li2019learning} overcame this lower bound in the setting in which $G$ is a graph by assuming that the graph learned is sampled from an Erd\H{o}s-R\'enyi model. We generalize the result of Li et al.\ to the setting of random $k$-uniform hypergraphs. To achieve this result, we leverage a novel equivalence between the problem of learning a single hyperedge and the standard group testing problem. This latter result may also be of independent interest.
\end{abstract}

\section{INTRODUCTION}
The problem of learning graphs through edge-detecting queries has been extensively studied due to its many applications, ranging from learning pairwise chemical reactions to genome sequencing problems \citep{AlonA05,AlonBKRS04,angluin2008learning,GrebinskiK98,ReyzinS07}. While significant progress has been made in this area, less is known about efficiently learning hypergraphs, which have now become the de facto standard to model higher-order network interactions \citep{battiston2020networks,benson2016higher,lotito2022higher}. In this paper, we take another step toward bridging this gap in the literature.

We study the problem of learning a hypergraph by making hyperedge-detecting queries. In particular, we focus on the \emph{non-adaptive setting} (in which every query needs to be submitted in advance), that is more suited for bioinformatics applications. 

A lower bound of \cite{abasi2019learning} shows that it is impossible in general to design algorithms that achieve a query complexity nearly linear in the number of (hyper)edges, even for graphs. A recent paper by \cite{li2019learning} shows that this lower bound can be beaten for graphs that are generated from a known Erd\H{o}s-R\'enyi model. We extend their results by providing algorithms for learning random hypergraphs that have nearly linear query complexity in the number of hyperedges.

\subsection{Background and Related Work}

\paragraph{From Group Testing to Hypergraph Query Learning} In the standard group testing model \cite{aldridge2019group,dorfman1943detection,du1999combinatorial}, one is given a finite set, containing a(n unknown) set of faulty elements. The main task of interest is to recover the set of faulty elements exclusively by repeatedly asking questions of the form:

\begin{center}
    {\em ``Does this subset contain at least one faulty element?''}
\end{center}

At its core, the problem of learning a hypergraph via hyperedge-detection queries is a constrained group testing problem. Here, the role of the faulty item is taken by the hyperedges of an unknown $k$-uniform hypergraph $G$ supported on a known set of vertices $V$. Note that, if one was allowed to ask whether an arbitrary collection $\mathcal{S}$ of elements of $\binom{V}{k}$ contains a hyperedge, then the problem would be entirely analogous to the standard group testing problem. Instead, we require that the collection of hyperedges queried is of the form $\mathcal{S} = \binom{S}{k}$ for some subset $S \subseteq V$. Intuitively, the fact that the queries must be specified by a subset of $V$, as opposed to a subset of $\binom{V}{k}$, renders the problem more difficult.

Recent advances in this model focus on algorithms that achieve low decoding time, and in this paper, we make use of a result of~\cite{cheraghchi2019simple} within a reduction used by one of our algorithms.

\paragraph{Learning Hypergraphs} \cite{torney1999sets} was the first to generalize group testing to the problem of learning hypergraphs via hyperedge-detection queries, a problem he refers to as testing for \emph{positive subsets}. The problem has since come under different names including \emph{group testing for complexes} \citep{chodoriwsky2015adaptive,macula2004trivial} and the \emph{monotone DNF query learning problem}~\citep{angluin1988queries,gao2006construction,abasi2014exact}.

\cite{angluin2008learning} show that learning arbitrary (non-uniform) hypergraphs of rank $k$ with $m$ hyperedges requires at least $\Omega\left(\left({2m\over k}\right)^{k/2}\right)$ queries.\footnote{We note that in general, one must require the hypergraph to be a \emph{Sperner} hypergraph, i.e. one in which no hyperedge is a subset of another, since otherwise the learning problem is not identifiable (See, e.g. \cite{abasi2018non}).} The same authors~\citep{angluin2006learning} showed that an arbitrary $k$-uniform hypergraph can be learned with high probability by making at most $O(2^{4k} m \operatorname{poly}(m,\log n))$ hyperedge-detection queries. Their algorithm makes use of $O(\min\{2^k(\log m + k)^2, (\log m +k)^3\})$ adaptive rounds. They also relax the uniformity condition by giving algorithms that perform well when the hypergraph is nearly uniform.

\cite{abasi2014exact} designed randomized adaptive algorithms for learning arbitrary hypergraphs with hyperedge-detecting queries. They also provide lower bounds for the problem they consider.

\cite{gao2006construction} gave the first explicit non-adaptive algorithm for learning $k$-uniform hypergraphs (exactly and with probability one) from hyperedge-detection queries. \cite{abasi2018non} then give non-adaptive algorithms for learning arbitrary hypergraphs of rank (at most) $k$ in the same setting that run in %
polynomial time in the optimal query complexity for their version of the problem. This, in general, may not be polynomial in the size of the hypergraph. \cite{abasi2018error} considers the same problem in the presence of errors. In particular, they focus on a model in which up to an $\alpha$-fraction of the queries made may return the incorrect answer.

\cite{balkanski2022learning} study algorithms for learning restricted classes of hypergraphs. They give an $O(\log^3 n)$-adaptive algorithm for learning an arbitrary hypermatching (a hypergraph with maximum degree 1) which makes $O(n\log^5 n)$ hyperedge-detection queries and returns the correct answer with high probability. %

\subsection{Our Results}

In this paper, we generalize the results of \cite{li2019learning} to Erd\H{o}s-R\'enyi hypergraphs.

In Section~\ref{sec:typical-instances}, we discuss a class of typical instances and use it to derive unconditional lower bounds on the learning problem. In Section~\ref{sec:grotesque}, we give an algorithm which solves the problem with low query complexity and decoding time. In particular, we prove the following:
\begin{restatable}{theorem}{grotesque}\label{thm:main-grotesque-theorem}
    There exists an algorithm (Algorithm~\ref{alg:grotesque}) that on input a hyperedge-detection oracle for an Erd\H{o}s-R\'enyi hypergraph, makes $O(k\bar{m} \log^2\bar{m} + k^2 \bar{m} \log \bar{m} \log^2 n)$ non-adaptive queries to the oracle, and outputs the correct answer with probability $\Omega(1)$. Here, the probability is taken over the randomness in both the algorithm and in the hypergraph. The algorithm requires $O(k\bar{m} \log^2\bar{m} + k^3 \bar{m} \log \bar{m} \log^2 n)$ decoding time.
\end{restatable}

In Section \ref{sec:lower-query-complexity-algorithms}, we will go over hypergraph adaptions of popular group testing algorithms. Specifically, we adapt the \texttt{COMP}, \texttt{DD} and \texttt{SSS} algorithms (see, e.g. \cite{aldridge2019group}), and establish that they all output the correct hypergraph with probability $\Omega(1)$ with $\Omega(\bar{m}\log n)$ queries, thus achieving a better query complexity than the algorithm in Theorem~\ref{thm:main-grotesque-theorem} at the price of a higher decoding time. %

\section{PRELIMINARIES}

\paragraph{Erd\H{o}s-R\'enyi Hypergraphs.} A $k$-uniform hypergraph is a tuple $G=(V,E)$ where $V$ is a finite set and $E \subseteq \binom{V}{k}$ is a collection of $k$-element subsets of $V$, called hyperedges. We refer to the elements of $V$ as \emph{nodes} or \emph{vertices} and denote by $n$ the number $|V|$ of vertices in $G$, and by $m$ the number $|E|$ of hyperedges. Whenever the hypergraph $G$ is not clear from context, we may use $m(G)$ to refer to the number of hyperedges in $G$. We refer to the cardinality $k$ of the hyperedges of $G$ as the \emph{rank} or the \emph{arity} of $G$.  While our guarantees have an explicit dependence on $k$, we will focus on the regime in which $k$ does not grow with $n$, i.e. $k=O(1)$.

For any hypergraph $G$, we define its maximum degree as:
\[
    \Delta(G) := \max_{v\in V}|\{h\in E\mid v\in h\}|.
\]

We will consider hypergraphs generated according to the Erd\H{o}s-R\'enyi model $G^{(k)}(n,q)$ in which every $k$-subset of $V$ is present with probability $q$. We denote by $\overline{m}$ the expected number of hyperedges in $G$ under this generative model, i.e.:
\[
    \overline{m} =q\binom{n}{k}.
\]
Note that, under this generative model, $m$ is a random variable, while $\overline{m}$, $n$ and $k$ are deterministic quantities.

\paragraph{Problem Setup.} This paper aims to build on the structure of \cite{li2019learning}, generalizing their results to hypergraphs. With that in mind, we briefly go over the setting; learning an unknown hypergraph generated via the Erd\H{o}s-R\'enyi model. We note that after generation, our hypergraph, $G$, remains fixed, and we try to uncover $G$ through a series of queries to an oracle. Where we ask if a set of vertices contains an edge. 

Upon completion of querying, the results create a decoder that forms an estimate of $G$, $\hat{G}$. The focus of this paper is to find algorithms minimizing the amount of queries, while maintaining the probability the decoder recovers $G$ be arbitrarily close to one. %

\paragraph{Sparsity Level.} As~\cite{li2019learning} limit their results to sparse graphs, we limit the scope of this paper to the standard notion of sparse hypergraphs in the following sense.
We assume that $q=o(1)$ as $n \rightarrow\infty$, and throughout this paper we will set $q=\Theta\left(n^{-k(1-\theta)}\right)$ for some $\theta \in(0,1)$, so that the average number of hyperedges $\bar{m}=\binom{n}{k}q$ behaves as $\Theta\left(n^{k \theta}\right)$. 
For efficient decoding results pertaining to Algorithm~\ref{alg:grotesque}, we use a stronger notion of
sparsity, where a superlinear number of edges 
are still allowed, but we further assume that $m = o( n^{k\over k-1})$.
We leave the question of tackling less sparse hypergraphs
open.

\paragraph{Bernoulli Random Queries.} We will often make use of Bernoulli queries, also known as Bernoulli tests. A Bernoulli query on a hypergraph $G=(V,E)$ is one in which the query is selected at random, by including each vertex $v\in V$ to be queried with a fixed probability $p$, independently of all other vertices. Following \cite{li2019learning}, we set $p=\sqrt[k]{\frac{k\nu}{qn^k}}$ for some constant $\nu >0$, and we note that this choice of $p$ gives $p^k=\frac{\nu}{m}(1+o(1))$, since $\bar{m}=\frac{1}{k} q n^k(1+o(1))$. Given a fixed hypergraph $G$ we will denote by $P_G$ the probability that a Bernoulli test with parameter $p$ as above is positive.

\section{TYPICAL INSTANCES}\label{sec:typical-instances}
In this section, we identify a set of typical instance arising from the random hypergraph model. This will allow us to make assumptions about the structure of the specific instance we are learning. We then use this to derive an information-theoretic lower bound on the query complexity of non-adaptively learning hypergraphs.

\begin{definition}[$\varepsilon$-typical Hypergraph Set]
For any $\varepsilon>0$, we define the $\varepsilon$-typical hypergraph set as the set $\mathcal{T}(\varepsilon)$ of hypergraphs $G$ satisfying both of the following conditions:
\begin{enumerate}
    \item $\left(1-\varepsilon\right) \bar{m} \leq m(G) \leq\left(1+\varepsilon\right) \bar{m}$,
    \item $\Delta(G) \leq d_{\max }$,
    \item $ \left(1-\varepsilon\right)\left(1-e^{-\nu}\right) \leq P_G \leq\left(1+\varepsilon\right)\left(1-e^{-\nu}\right)$.
\end{enumerate}

where:
$$
d_{\max }= \begin{cases}k n^{k-1} q & \theta>\frac{1}{k} \\ \log n & \theta \leq \frac{1}{k} .\end{cases}
$$
\end{definition}

We now show that, for any $\varepsilon >0$, $
\Pr[G \in T(\varepsilon)] \to 1$ as $n \to \infty$, where the probability is taken over the random choice of $G$ from $G^{(k)}(n,q)$. This key result is a hypergraph analogue of a similar result appearing in the paper of~\cite{li2019learning}.

\begin{lemma}\label{lem:typical-set}
    For any $\varepsilon > 0$, we have:
    \[
        \Pr[G \in T(\varepsilon)] \to 1
    \]
    as $n \to \infty$.
\end{lemma}
\begin{proof}
We begin by noting that the set $\mathcal{T}(\varepsilon)$ can be written as $\mathcal{T}(\varepsilon) = \mathcal{T}^{(1)}(\varepsilon) \cap \mathcal{T}^{(2)}(\varepsilon) \cap \mathcal{T}^{(3)}(\varepsilon)$, where:
\begin{align*}
    \mathcal{T}^{(1)} &= \{G : \left(1-\varepsilon\right) \bar{m} \leq m(G) \leq\left(1+\varepsilon\right) \bar{m}\}\\
    \mathcal{T}^{(2)} &= \{G : \Delta(G) \leq d_{\max }\}\\
    \mathcal{T}^{(3)} &= \{G : \left(1-\varepsilon\right)\left(1-e^{-\nu}\right)\leq P_G 
    \leq\left(1+\varepsilon\right)\left(1-e^{-\nu}\right)\}
\end{align*}

It is then sufficient to show that $\Pr[G \in \mathcal{T}^{(i)}] \to 1$ for every $i =1, 2, 3$.

Since $m(G)$ follows a binomial distribution with parameters $\binom{n}{k}$ and $q$, we have:
\[
    \Pr[(1-\varepsilon)\bar{m} \leq m(G) \leq (1+\varepsilon) \bar{m}] \to 1
\]
as $\binom{n}{k}\to\infty$. This yields $\Pr[G \in \mathcal{T}^{(1)}] \to 1$.

We now establish that $\Pr[G \in \mathcal{T}^{(2)}] \to 1$: 

\begin{enumerate}
    \item If $\theta>\frac{1}{k}$, then the combinatorial degree of each vertex follows a binomial distribution with mean $\binom{n-1}{k-1}q=\Theta(n^c)$ for some $c>0$. We can then follow the work \cite{li2019learning}, using the Chernoff bound to show the probability of any degree exceeding $kn^{k-1}q$ goes to zero.

    \item If $\theta\le \frac{1}{k}$, we note that we need only consider the case $\theta =\frac{1}{k}$, as this is when the probability for exceeding $\log n$ degree is highest. Here, we have that the combinatorial degree for a vertex follows a binomial distribution with $\binom{n-1}{k-1}$ trials and success probability $\Theta(\frac{1}{n^{k-1}})$, so the mean is $\Theta(1)$. From here we can once again follow the argument of \cite{li2019learning}, using the standard Chernoff bound to show that the probability of any vertex exceeding $\log n$ degree vanishes. 
\end{enumerate}

The last and most intensive argument is to establish $\mathcal{T}^{(3)}(\varepsilon)$. However, the hypergraph extension of this result is straightforward, we simply adapt the proof in the paper of \cite{li2019learning}, making note that the constant, two, used in the graph case becomes $k$ in our new hypergraph setting.
\end{proof}

A simple consequence of Lemma~\ref{lem:typical-set} is that the algorithm-independent lower bound for the number of tests needed to obtain asymptotically vanishing probability provided in~\cite{li2019learning} holds for general hypergraphs.%

\begin{restatable}{theorem}{LowerBoundTheorem}\label{thm:lower-bound-theorem}
    Under the typical instance setting discussed above, with $q=o(1)$ and an arbitrary non-adaptive test design, to have vanishing error probability we must have at least $\left(\bar{m} \log _2 \frac{1}{q}\right)(1-\eta)$ queries, for arbitrarily small $\eta>0$.
\end{restatable}
For completeness, we include the full proof of Theorem~\ref{thm:lower-bound-theorem} in Section \ref{sec:proof-of-lower-bound} of the Supplementary Materials.

\section{THE \texttt{HYPERGRAPH-GROTESQUE} ALGORITHM}\label{sec:grotesque}
In this section we give a sublinear-time decoding algorithm for the problem of learning hypergraphs with hyperedge detection queries. As in the previous sections, we assume that the hypergraph is sampled according to the Erd\H{o}s-R\'enyi model, and the probabilistic guarantees of the algorithm will depend on the randomness in both the algorithm and the hypergraph generative process.

We prove the main theorem:
\grotesque*
\begin{algorithm}[H]
    \caption{\texttt{HYPERGRAPH-GROTESQUE}}\label{alg:grotesque}
    \begin{algorithmic}
        \State \textbf{Input:} A hyperedge-detection oracle for a hypergraph~$G$.
        \State \textbf{Output:} A hypergraph $\hat{G} = (V,\hat{E})$.\vspace{2mm}
        \State Let $b = \Theta(\bar{m} \log \bar{m})$ be given as in Section~\ref{ssec:bundles}.
        \State Form bundles $B_1, ..., B_b$ by independently including each vertex $v$ in each bundle $B_i$ with probability $r_{inc} = {1\over \sqrt[k]{2m}}$.
        \State Let $\delta^* \gets O({1 \over \bar{m}\log \bar{m}})$
        \State Initialize $\hat{E} = \emptyset$.
        \For{$i = 1, ... , b$} 
            \If{Multiplicity\_Test$(B_i,\delta^*)$ returns 1}
                \State Perform a location test (Section~\ref{ssec:location-tests}) on $B_i$, and add the resulting hyperedge $h$ to $\hat{E}$. 
            \EndIf
        \EndFor
        \State \textbf{Return} $\hat{G} = (V,\hat{E})$
    \end{algorithmic}
\end{algorithm}

\begin{algorithm}
    \caption{Multiplicity Test}\label{alg:mult-test}
    \begin{algorithmic}
    \State \textbf{Input:} A bundle $B \subseteq V$, an error probability $\delta$.
    \State \textbf{Output:} An outcome in $\{0,1\}$ indicating whether $\mathcal{B}$ contains a single hyperedge. \vspace{2mm}
    \State Let $M= {1\over e}(1-{1 \over \sqrt[k]{e}})$.
    \State Perform $t_{mul} = 2\log (2/\delta)/ M^2$ edge detection queries on $B$ chosen according to a Bernoulli design with parameter $r_{mul} = 1 / \sqrt[k]{e}$. Let $\hat{p}$ be the fraction of queries that return a positive outcome (i.e. the ones for which the set being queried contains a full hyperedge).
    \If{ $\hat{p} \in (0,1/e+M/2)$}
        \State \textbf{Return }1
    \Else
    \State \textbf{Return }0
    \EndIf
    \end{algorithmic}
\end{algorithm}

Similarly to the algorithm of \cite{li2019learning}, our algorithm is inspired by the \texttt{GROTESQUE} procedure first introduced by \cite{cai2017efficient}. In particular, the algorithm is structured according to the following high-level framework:
\begin{enumerate}
    \item In the first step, the algorithm produces random sets of vertices (bundles), obtained by including each vertex in each set independently with a fixed probability. This step is successful if each hyperedge is the unique hyperedge in at least one of the bundles. By a coupon-collector argument, one can bound from below the probability of this step succeeding when the number of bundles is sufficiently large (Section~\ref{ssec:bundles}).
    \item Then, the algorithm performs \emph{multiplicity tests} on each of the sets to identify the ones that contain a unique hyperedge. This works by estimating the probability of a Bernoulli test detecting a hyperedge within a bundle $\mathcal{B}$, and then using this estimate to determine whether the bundle really contains a single hyperedge. This step is successful if every multiplicity test correctly identifies whether a bundle contains a single hyperedge. By applying standard sampling results, it can be shown that, if sufficiently many Bernoulli tests are made, this step is successful with high probability (Section~\ref{ssec:multiplicty-test}). 
    \item Finally, the algorithm performs a \emph{location test} on the sets that passed the multiplicity test, which identifies the unique hyperedge the set contains. This step is successful if every location test correctly identifies the unique hyperedge in a bundle. We show that this step can be performed by leveraging a reduction to the standard group testing problem (Section~\ref{ssec:location-tests}).
\end{enumerate} 

It is not hard to see that if all three steps are successful, one can reconstruct the hypergraph $G$ correctly from the result of the queries.

We note that, while the procedure above is described sequentially, all of the tests needed to carry it out can be performed non-adaptively. 

We will now analyze each step in detail. After that, we complete the proof of Theorem~\ref{thm:main-grotesque-theorem}.

\subsection{Bundles of Tests}\label{ssec:bundles}
Recall that Algorithm~\ref{alg:grotesque} forms a number $b = \Theta(\bar{m}\log \bar{m})$ of bundles of vertices, where each node is placed independently in each bundle with probability $r_{inc} := 1/{\sqrt[k]{2m}}$. 

We say a hyperedge is \emph{fully contained} in a bundle $B$ if all of the vertices in the hyperedge have been placed in $B$. Intuitively, the random process of forming the bundles is successful if, for every hyperedge $h$, there exists a bundle ${B}_i$ such that $h$ is the \emph{unique} hyperedge that is fully contained in ${B}_i$.  We prove the following lemma:

\begin{restatable}{lemma}{LemmaHyperedgeUnique}\label{lemma:hyperedge-is-unique-in-bundle}
    Let $G$ be any $k$-uniform hypergraph. Suppose that the vertices of $G$ are placed into bundles according to the procedure described in Algorithm~\ref{alg:grotesque}. For any fixed hyperedge $h \in G$ and any fixed bundle $B_i$, let $\mathcal{E}_{h,i}$ be the event that $h$ is the only hyperedge fully contained in $B_i$. Then:
    \[
        \Pr[\mathcal{E}_{h,i}] \geq \left(1- r_{inc}k\Delta(G)- m(G)r^k\right) r_{inc}^k.
    \]
\end{restatable}
\begin{proof}
We consider three events $A_0$, $A_1$ and $A_2$ defined as follows:
\begin{itemize}
    \item $A_0$ is the event that the hyperedge $h$ is fully contained in $B_i$,
    \item $A_1$ is the event that there exists some hyperedge $h'\neq h$ satisfying $h' \cap h \neq \emptyset$ that is fully contained in $B_i$,
    \item $A_2$ is the event that there exists some hyperedge $h'$ satisfying  $h' \cap h = \emptyset$ that is fully contained within $B_i$.
\end{itemize}
By definition, we have $\mathcal{E}_{h,i} = A_0 \cap \overline{A}_1 \cap \overline{A}_2$. Note that:
\[
    \Pr[A_0] = r_{inc}^k,
\]
while by the union bound, we have:
\[
    \Pr[{A_1} \mid A_0] \leq \sum_{\substack{h'\in E\setminus\{h\}\\ h' \cap h \neq \emptyset}} r_{inc} \leq r_{inc}k\Delta(G),
\]
and:
\[
    \Pr[{A_2} \mid A_0] = \Pr[A_2] \leq \sum_{\substack{h' \in E\\ h'\cap h = \emptyset}} r_{inc}^k = m(G) r_{inc}^k. 
\]
We then have:
\begin{align*}
    \Pr[\mathcal{E}_{h,i}] &= \Pr[A_0 \cap \bar{A}_1 \cap \bar{A}_2] = \Pr[\bar{A}_1 \cap \bar{A}_2\mid A_0] \Pr[A_0]\\
    &= \left(1- \Pr[{A}_1 \cup {A}_2\mid A_0]\right) r_{inc}^k\\
    &\geq \left(1- \Pr[{A}_1\mid A_0] - \Pr[{A}_2\mid A_0]\right) r_{inc}^k\\
    &=\left(1- r_{inc}k\Delta(G)- m(G)r^k\right) r_{inc}^k.
\end{align*}
\end{proof}

This in turn gives the following result.

\begin{lemma}\label{lem:bundles-guarantees} 
When the \texttt{HYPERGRAPH-GROTESQUE} algorithm is run on a hypergraph $G$ sampled according to an Erd\H{o}s-R\'enyi model for sufficiently large values of $n$, the probability that every hyperedge $h$ is the unique hyperedge in some bundle of tests satisfies:
\[
    \Pr\left[\bigcap_{h\in E}\bigcup_{i\in[b]} \mathcal{E}_{h,i}\right] \geq 1-\delta.
\] 
\end{lemma}
\begin{proof}
Let $m = m(G)$. For every fixed $h$ and $i$, by Lemma~\ref{lemma:hyperedge-is-unique-in-bundle}:
\begin{align}\label{eq:bound-on-Ehb}
     \Pr[\mathcal{E}_{h,i}] &\geq \left(1- r_{inc}k\Delta(G)- mr_{inc}^k\right) r_{inc}^k \nonumber \\
     &= \left(1- r_{inc}{k^2m\over n}- {1\over 2}\right) {1\over 2m} \nonumber \\ 
     &= {1\over 4m} (1+o(1)),
\end{align}
where we are using the fact that $n = \omega(m^{1-{1\over k}})$. Hence:
\begin{align*}
    \Pr\left[\overline{\bigcap_{h\in E}\bigcup_{i\in[b]} \mathcal{E}_{h,i}}\right] &= \Pr\left[\bigcup_{h\in E}\bigcap_{i\in[b]} \overline{\mathcal{E}_{h,i}}\right] \\
    &\leq  \sum_{h\in E} \Pr\left[\bigcap_{i\in[b]} \overline{\mathcal{E}_{h,i}}\right] \\
    &= \sum_{h\in E} \prod_{i\in[b]} \Pr\left[\,\overline{\mathcal{E}_{h,i}}\,\right] \\
    &\leq \sum_{h\in E} \prod_{i\in[b]} \left(1-{1\over 4m} (1+ o(1))\right) \\
    &= m \left(1-{1\over 4m} (1+ o(1))\right)^b \\
    &\leq m e^{-{b\over 4m} (1+o(1))},
\end{align*}
where we first applied De Morgan's law, then the union bound, then the fact that for every $h$ the random variables $\{\mathcal{E}_{h,i}\}_{i\in[b]}$ are mutually independent, then Equation~\eqref{eq:bound-on-Ehb}. The result then follows.
\end{proof}

\subsection{Multiplicity Test}\label{ssec:multiplicty-test}
We now discuss the guarantees of the multiplicity test.
\begin{definition}
    Given a set $B \subseteq V$, a $(r_{mul},t_{mul})$-multiplicity test for $B$ is a collection of $t_{mul}$ tests on the elements of $B$ chosen according to a Bernoulli design with parameter $r_{mul}$. The test returns $1$ if the fraction of positive tests suggests that a single hyperedge is present in the bundle and $0$ otherwise.
\end{definition}

In order to analyze the multiplicity test (Algorithm~\ref{alg:mult-test}), we use the following lemma, which we prove in Section~\ref{sec:missing-proofs-grotesque} of the Supplementary Materials.
\begin{restatable}{lemma}{detectionProbability}\label{lem:detection-probability}
    Suppose that a set $B \subseteq V$ contains multiple hyperedges. Let $S$ be a subset chosen according to a Bernoulli design with parameter $1/\sqrt[k]{e}$ (i.e. by including each $v \in B$ into $S$ independently with probability $1/\sqrt[k]{e}$). Then the probability that $S$ contains a full hyperedge is at least:
    $
        {2 / e} - {1 / e^{(k+1)/k}}.
    $
\end{restatable}

This then yields the following guarantee on correctness, also proved in Section~\ref{sec:missing-proofs-grotesque}:
\begin{restatable}{lemma}{MultiplicityTestIsCorrect}\label{lemma:multiplicity-test-is-correct}
    Suppose we run a multiplicity test on a bundle $B$ with error probability parameter $\delta$. Then:
    \begin{enumerate}
        \item if $B$ contains no hyperedge, the answer is always $0$,
        \item if $B$ contains a single hyperedge, the test returns $1$ with probability at least $1-\delta$,
        \item and if $B$ contains more than one hyperedge, the test returns $0$ with probability at least $1-\delta$.
    \end{enumerate}  
\end{restatable}

In the same section, we also obtain the following guarantee on the efficiency of the multiplicity tests:
\begin{restatable}{lemma}{MultiplicityTestIsEfficient}\label{lemma:multiplicity-test-is-efficient}
    The number of queries made by a multiplicity test with error parameter $\delta$ is at most $e^3 k \log {2\over \delta}$, and the decoding time for each multiplicity test is $O(k \log {1\over \delta})$.
\end{restatable}

\subsection{Location Test via Reduction to Group Testing}\label{ssec:location-tests}
Once the algorithm has performed all the multiplicity tests, it runs location tests on the bundles that passed the multiplicity tests. Executing a location test on a bundle that contains a single hyperedge $h$ allows the algorithm to \emph{discover} $h$ and add it to the estimate hypergraph $\hat{H}$. 

We obtain a location test by highlighting an equivalence between the problem of learning a single hyperedge of arity $k$ using a hyperedge-detection oracle, and that of group testing with $k$ defective items.

\begin{lemma}\label{lem:location-test-is-group-test}
    Any algorithm for the group testing problem with $k$ faulty items yields an algorithm for the problem of learning a hypergraph known to have a single hyperedge of arity $k$ by making edge-detection queries. Conversely, any algorithm for the latter problem yields an algorithm for the former.
\end{lemma}

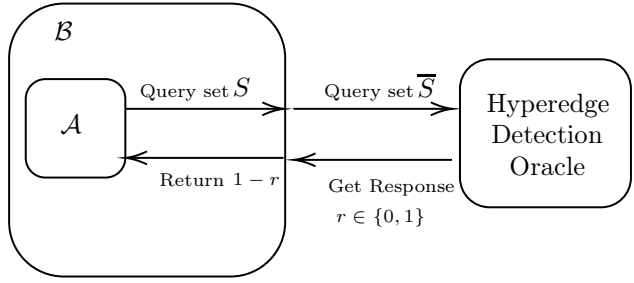
\begin{figure}
    \centering
    \tikzset{every picture/.style={line width=0.75pt}} %

\begin{tikzpicture}[x=0.75pt,y=0.75pt,yscale=-1,xscale=1]
\draw   (79,88.94) .. controls (79,83.45) and (83.45,79) .. (88.94,79) -- (119.29,79) .. controls (124.78,79) and (129.23,83.45) .. (129.23,88.94) -- (129.23,118.76) .. controls (129.23,124.25) and (124.78,128.7) .. (119.29,128.7) -- (88.94,128.7) .. controls (83.45,128.7) and (79,124.25) .. (79,118.76) -- cycle ;

\draw   (70.62,68.43) .. controls (70.62,53.22) and (82.95,40.88) .. (98.17,40.88) -- (182.3,40.88) .. controls (197.51,40.88) and (209.85,53.22) .. (209.85,68.43) -- (209.85,151.09) .. controls (209.85,166.3) and (197.51,178.64) .. (182.3,178.64) -- (98.17,178.64) .. controls (82.95,178.64) and (70.62,166.3) .. (70.62,151.09) -- cycle ;
\draw    (129,94) -- (207.23,94) ;
\draw [shift={(209.23,94)}, rotate = 180] [color={rgb, 255:red, 0; green, 0; blue, 0 }  ][line width=0.75]    (10.93,-3.29) .. controls (6.95,-1.4) and (3.31,-0.3) .. (0,0) .. controls (3.31,0.3) and (6.95,1.4) .. (10.93,3.29)   ;
\draw    (209,118.76) -- (131.23,118.76) ;
\draw [shift={(129.23,118.76)}, rotate = 360] [color={rgb, 255:red, 0; green, 0; blue, 0 }  ][line width=0.75]    (10.93,-3.29) .. controls (6.95,-1.4) and (3.31,-0.3) .. (0,0) .. controls (3.31,0.3) and (6.95,1.4) .. (10.93,3.29)   ;
\draw   (298,83.94) .. controls (298,75.69) and (304.69,69) .. (312.94,69) -- (369.29,69) .. controls (377.54,69) and (384.23,75.69) .. (384.23,83.94) -- (384.23,128.76) .. controls (384.23,137.01) and (377.54,143.7) .. (369.29,143.7) -- (312.94,143.7) .. controls (304.69,143.7) and (298,137.01) .. (298,128.76) -- cycle ;

\draw    (214,94) -- (292.23,94) ;
\draw [shift={(294.23,94)}, rotate = 180] [color={rgb, 255:red, 0; green, 0; blue, 0 }  ][line width=0.75]    (10.93,-3.29) .. controls (6.95,-1.4) and (3.31,-0.3) .. (0,0) .. controls (3.31,0.3) and (6.95,1.4) .. (10.93,3.29)   ;
\draw    (294,119.76) -- (216.23,119.76) ;
\draw [shift={(214.23,119.76)}, rotate = 360] [color={rgb, 255:red, 0; green, 0; blue, 0 }  ][line width=0.75]    (10.93,-3.29) .. controls (6.95,-1.4) and (3.31,-0.3) .. (0,0) .. controls (3.31,0.3) and (6.95,1.4) .. (10.93,3.29)   ;

\draw (95,95.4) node [anchor=north west][inner sep=0.75pt]    {$\mathcal{A}$};
\draw (182,77) node [anchor=north west][inner sep=0.75pt]    {$S$};
\draw (135,80) node [anchor=north west][inner sep=0.75pt]  [font=\scriptsize] [align=left] {Query set};
\draw (92,48.4) node [anchor=north west][inner sep=0.75pt]    {$\mathcal{B}$};
\draw (310,78) node [anchor=north west][inner sep=0.75pt]   [align=left] {\begin{minipage}[lt]{45.83pt}\setlength\topsep{0pt}
\begin{center}
Hyperedge \\Detection\\Oracle
\end{center}

\end{minipage}};
\draw (275,75) node [anchor=north west][inner sep=0.75pt]    {$\overline{S}$};
\draw (228,80) node [anchor=north west][inner sep=0.75pt]  [font=\scriptsize] [align=left] {Query set};
\draw (230,128) node [anchor=north west][inner sep=0.75pt]  [font=\scriptsize] [align=left] {Get Response};
\draw (234,142.4) node [anchor=north west][inner sep=0.75pt]  [font=\scriptsize]  {$r\in \{0,1\}$};
\draw (145,125) node [anchor=north west][inner sep=0.75pt]  [font=\scriptsize] [align=left] {Return};
\draw (182,125) node [anchor=north west][inner sep=0.75pt]  [font=\scriptsize]  {$1-r$};

\end{tikzpicture}
    \caption{The structure of the reduction in the proof of Lemma~\ref{lem:location-test-is-group-test}. Here, the algorithm $\mathcal{B}$ is given access to a hyperedge-detection oracle. $\mathcal{B}$ simulates algorithm $\mathcal{A}$ and converts $\mathcal{A}$'s queries into hyperedge detection queries.}
    \label{fig:reduction-diagram}
\end{figure}
\begin{proof}
Consider the following reduction from the latter problem to the former. Suppose $\mathcal{A}$ is an algorithm that solves the group testing problem: i.e. given a finite set $V$ which contains a subset $K$ of defective items, $\mathcal{A}$ submits queries of the form $S \subseteq V$ to an oracle that determines whether $S\cap K \neq \emptyset$, and based on the answer to those queries, it recovers $K$.

Now, consider the problem of learning a cardinality-$k$ hyperedge $h$ on $V$ by making hyperedge-detection queries. We design an algorithm $\mathcal{B}$ for the latter problem as follows: $\mathcal{B}$ simulates $\mathcal{A}$ and whenever $\mathcal{A}$ submits a query $S \subseteq V$, $\mathcal{B}$ instead submits the query $\overline{S}$ to the hyperedge-detection oracle, and then returns to $\mathcal{A}$ the opposite ($1-r$) of the answer $r$ it receives. When $\mathcal{A}$ terminates, outputting a set $S^*$, $\mathcal{B}$ outputs the same set.

For each query $S$ made by $\mathcal{A}$ the value of $1-r$ is equal to 1 if and only if the set $S$ contains at least one element of the hidden hyperedge $h$. Hence, from the perspective of $\mathcal{A}$, $\mathcal{B}$ is implementing a group testing oracle for an instance in which $K=h$. In particular, if $\mathcal{A}$ correctly solves the group testing problem, the output $S^*$ of $\mathcal{B}$ is equal to $h$.

It is easy to see that an analogous reduction can be used to reduce from the group testing problem to that of learning a single hyperedge, and hence the two problems are entirely equivalent.
\end{proof}

Note that this reduction preserves query complexity, adaptivity, and runtime guarantees.

The group testing problem is well-studied in the literature and in particular the following result is known.

\begin{theorem}[Paraphrasing Theorem 11 from \cite{cheraghchi2019simple}]
    Consider the standard group testing problem on $n$ element with $k$ defective elements. There exists a(n explicitly constructable) collection of $O(k^2 \log^2 n)$ group tests and an algorithm $\mathcal{A}$ which, given the results of the tests as input, outputs the set of defective items in $O(k^3 \log^2 n)$ time.
\end{theorem}

The result of Cheraghchi and Ribeiro is based on a construction of linear codes with fast decoding time described in the same paper.

By Lemma~\ref{lem:location-test-is-group-test}, the above result implies:
\begin{corollary}\label{cor:good-location-test-exists}
    Consider the problem of learning a hypergraph known to consist of a single hyperedge of arity $k$ by non-adaptively making queries to a hyperedge-detection oracle. There exists an algorithm for this problem which makes $O(k^2\log^2 n)$ queries and requires decoding time $O(k^3 \log^2 n)$.
\end{corollary}
The algorithm guaranteed by Corollary~\ref{cor:good-location-test-exists} is simply the group testing algorithm of~\cite{cheraghchi2019simple} run through the reduction used in the proof of Lemma~\ref{lem:location-test-is-group-test}.

\subsection{Proof of Theorem~\ref{thm:main-grotesque-theorem}}

\begin{proof}[{\bf Proof of Theorem~\ref{thm:main-grotesque-theorem}}]
    By Lemma~\ref{lem:bundles-guarantees}, if we create $b = \Theta(\bar{m}\log \bar{m})$ bundles of vertices, and assign each vertex to a bundle at random with probability $r_{inc}={1/\sqrt[k]{2\bar{m}}}$, every hyperedge is the unique hyperedge in some bundle with constant probability. 
    
    The algorithm then runs a multiplicity test with error probability $\delta^* = \Theta\left({1\over \bar{m}\log \bar{m}}\right)$ on every bundle. By Lemma~\ref{lemma:multiplicity-test-is-correct} and the union bound, there is a constant probability that every multiplicity test succeeds. 
    
    By Lemma~\ref{lemma:multiplicity-test-is-efficient}, this requires $O(k\log \bar{m} )$ queries and $O(k\log \bar{m})$ decoding time for every bundle, which amounts to a total of $O(k\bar{m}\log^2\bar{m})$ queries and decoding time to establish which bundles contain a single hyperedge. We then need to run $\bar{m}\log \bar{m}$ location tests\footnote{Note that prior to running a location test on a bundle $B$, the algorithm can check whether $B$ contains a previously discovered hyperedge $h$. If that is the case (i.e.\ if $h \in B$), then $B$ could be ignored, and the algorithm would not run a location test on it. This allows one to guarantee that in a successful run of the algorithm, no more than $m$ location tests are run. However performing this check comes at an extra computational cost, and it is not clear that it can be carried out efficiently. In the paper of \cite{li2019learning} the authors do not discuss this issue and rather just assume that one is able to run at most $m$ location tests through the run of the algorithm. By doing this, they remove the factor of $\log m$ from the second term in the above bound}, each of which requires $O(k^2 \log^2 n)$ queries and $O(k^3 \log^2 n)$ decoding time. The total then equals:
    \[
        O(k\bar{m} \log^2\bar{m} + k^2 \bar{m} \log \bar{m} \log^2 n ))
    \]
    queries and:
    \[
        O(k\bar{m} \log^2\bar{m} + k^3 \bar{m} \log \bar{m} \log^2 n))
    \]
    decoding time as needed (Recall that $m = \Theta(\bar{m})$ with probability that tends to $1$ as $n \to \infty$).
\end{proof}

\section{OTHER ALGORITHMIC RESULTS}\label{sec:lower-query-complexity-algorithms}

This section presents hypergraph analogues of popular group testing algorithms, building upon the results given by \cite{li2019learning} in the context of graphs. We also provide formal guarantees on the query complexity and success probability of the algorithms we describe, showing that these algorithms have a better query complexity ($O(k\bar{m} \log n)$) than Algorithm~\ref{alg:grotesque} (at the price of longer decoding time). We defer the proofs of the results in this section to Section~\ref{sec:missing-proofs-group-testing} of the Supplementary Materials.

The three algorithms we adapt are ``Combinatorial Orthogonal Matching Pursuit'' \texttt{COMP}, ``Definite Defectives'' \texttt{DD}, and ``Smallest Satisfying Set'' \texttt{SSS}. The \texttt{COMP} algorithm for group testing simply rules out all of the elements that have appeared in any negative tests and returns the remaining elements. The \texttt{DD} algorithm, first rules out all elements that appear in any negative test, then outputs all the elements that must be defective out of the remaining elements. \texttt{SSS} simply returns a satisfying assignment of minimum cardinality. We refer the reader to the survey of \cite{aldridge2019group} for a review of how these algorithms are used in group testing.

All three of these algorithms produce an estimate $\hat{G}$ of the hypergraph $G$ based on the result of a single batch of Bernoulli queries. In particular, we will assume that each algorithm takes as input a collection $\{X^{(i)}\}_{i\in [t]}$ of hyperedge-detection queries, where each $X^{(i)} \subseteq V$ is chosen according to a Bernoulli design with parameter $p = \sqrt[k]{k\nu\over qn^k}$ (for some $\nu$ to be defined). We also assume the algorithms have access to the results $\{Y^{(i)}\}_{i\in [t]}$ of the queries, where:
\[
    Y^{(i)} =\begin{cases}
        1 &\text{if there exists }h\in H \text{ s.t. }h\subseteq X^{(i)}\\
        0 & \text{otherwise.}
    \end{cases}
\]
Since the algorithms themselves are deterministic, all of the probabilistic guarantees are based on the randomness in both the choice of Bernoulli queries and the hypergraph generation process.

\paragraph{The \texttt{COMP} Algorithm.} The first algorithm we examine is \texttt{COMP} (Algorithm~\ref{alg:COMP}). The key observation behind this algorithm is the following: no collections $h$ of $k$ vertices can be a hyperedge in $G$ if all the vertices in $h$ appear in some query $X^{(i)}$ with $Y^{(i)} =0$. The algorithm then simply assumes each hyperedge $h$ is present in $G$ unless it satisfies this condition.

\begin{algorithm}
   \caption{\texttt{COMP}}\label{alg:COMP}
   \begin{algorithmic}
        \State \textbf{Input:} A hyperedge-detection oracle for a hypergraph $G$, $t$ hyperedge-detection queries $\{X^{(1)},\dots,X^{(t)}\}$, and the results $\{Y^{(1)},\dots,Y^{(t)}\}$ of the queries.
    
        \State \textbf{Output:} A hypergraph $\widehat{G}=(V,\widehat{E})$.
    
        \State \textbf{Initialize} $\widehat{E}$ to contain all $\binom{n}{k}$ edges
            \For {each $i$ such that $Y^{(i)}=0$}%
    
                \State Remove all $h$ from $\widehat{E}$ satisfying $h\subseteq X^{(i)}$\;
            \EndFor
    \State \textbf{return} $\widehat{G}=(V, \widehat{E})$
    \end{algorithmic}
\end{algorithm}
We obtain the following guarantees on the performance of \texttt{COMP}. 
\begin{restatable}{theorem}{TheoremCOMP}\label{thm:COMP-guarantees}
    If \texttt{COMP} is given as input an unknown hypergraph $G$ sampled from $G^{(k)}(n,q)$ that is in the typical instance setting, where we have $q=\Theta\left(n^{k(\theta-1)}\right)$ for some $\theta \in (0,1)$, as well as at least $t=k e \cdot \bar{m} \log n$ Bernoulli queries with parameter $\nu=1$, then it outputs $\hat{G} = G$ with probability $\Omega(1)$.
\end{restatable}

\paragraph{The \texttt{DD} Algorithm.} \texttt{COMP}'s method of assuming edges are present until proven otherwise may be rather inefficient since we are looking for a sparse graph 
The \texttt{DD} algorithm reverses this assumption and starts with all edges as non-edges, making use of \texttt{COMP} to preclude non-edges. 

\begin{algorithm}
   \caption{\texttt{DD}}\label{alg:DD}
   \begin{algorithmic}       
        \State \textbf{Input:} A hyperedge-detection oracle for a hypergraph $G$, $t$ sets of vertices, $\{X^{(1)},\dots,X^{(t)}\}$, to be queried, with oracle given binary responses, $\{Y^{(1)},\dots,Y^{(t)}\}$.
        \State \textbf{Output:} A hypergraph $\widehat{G}=(V,\widehat{E})$.
        \State\textbf{Initialize} $\widehat{E}=\emptyset$, and initialize a potential edge set, $\mathrm{PE}$, to contain all $\binom{n}{k}$ edges
        \For{ each $i$ such that $Y^{(i)}=0$ }
            \State Remove all edges from PE whose nodes are all in $X^{(i)}$
        \EndFor
        \For{ each $i$ such that $Y^{(i)}=1$}
        
            \State If the nodes from $X^{(i)}$ cover exactly one edge in $\mathrm{PE}$, add that edge to $\widehat{E}$\;
        \EndFor
        \State \textbf{return} $\widehat{G}=(V, \widehat{E})$.
   \end{algorithmic}
\end{algorithm}

\begin{restatable}{theorem}{TheoremDD}\label{thm:DD-guarantees}
If we have an unknown Erd\H{o}s-R\'enyi hypergraph that is in the typical instance setting, where we have $q=\Theta\left(n^{k(\theta-1)}\right)$ for some $\theta \in(0,1)$, and Bernoulli testing with parameter $\nu=1$, then with at least $k \max \{\theta, 1-\theta, 1-\theta/2, 1+\theta/2-1/k\} e \cdot \bar{m} \log n$ non-adaptive queries \texttt{DD} outputs the correct answer with probability $\Omega(1)$.
\end{restatable} 

\paragraph{The \texttt{SSS} Algorithm.}
The \texttt{SSS} algorithm works by finding the smallest set of edges such that the output is consistent with the Bernoulli test results, i.e. $\{Y^{(i)}\}_{i\in [t]}$. Since \texttt{SSS} searches for the minimal satisfying graph, it gives a lower bound to the size of the output of any Bernoulli-queries-based decoding algorithm. 

\begin{algorithm}
   \caption{\texttt{SSS}}\label{alg:SSS}
   \begin{algorithmic}       
        \State \textbf{Input:} A hyperedge-detection oracle for a hypergraph $G$, $t$ sets of vertices, $\{X^{(1)},\dots,X^{(t)}\}$, to be queried, with oracle given binary responses, $\{Y^{(1)},\dots,Y^{(t)}\}$.
        \State \textbf{Output:} A hypergraph $\widehat{G}=(V,\widehat{E})$.
        \State Find $\widehat{E}$ such that $|\widehat{E}|$ is minimized while satisfying $\{Y^{(1)},\dots,Y^{(t)}\}$ 
        \State \textbf{return} $\widehat{G}=(V, \widehat{E})$.
   \end{algorithmic}
\end{algorithm}

\begin{restatable}{theorem}{TheoremSSS}\label{thm:SSS-guarantees}
    If we have an unknown Erd\H{o}s-R\'enyi hypergraph that is in the typical instance setting, where we have $q=\Theta\left(n^{k(\theta-1)}\right)$ for some $\theta \in(0,1)$, and Bernoulli testing with an arbitrary choice of $\nu>0$, then with at least $k \theta e \cdot \bar{m} \log n$ non-adaptive queries the \texttt{SSS} algorithm outputs the correct answer with probability $\Omega(1)$. 
\end{restatable}

\section{OPEN PROBLEMS} The main open problem remains to improve the sparsity level of the low decoding time hypergraph learning \texttt{HYPERGRAPH-GROTESQUE} or to show that the sparsity assumption is necessary. Another direction is to improve its decoding time, which seems very likely to at least be possible with respect to logarithmic factors.

\newpage 
\bibliographystyle{abbrvnat}
\bibliography{bibliography}

\newpage
\onecolumn
\appendix
\vspace{-2.5cm}
\section{PROOFS FOR SECTION \ref{sec:typical-instances}}\label{sec:proof-of-lower-bound}
\LowerBoundTheorem*
\begin{proof}
We have the following entropy inequality from \cite{li2019learning}: 
    $$
\begin{aligned}
P_{\mathrm{e}}  \geq \mathbb{P}[\mathcal{A}] \frac{H(G \mid \mathcal{A}=\text { true })-I(G ; \widehat{G} \mid \mathcal{A}=\text { true })-\log 2}{\log \left|\mathcal{G}_{\mathcal{A}}\right|},
\end{aligned}
$$
where $P_{\mathrm{e}}$ is the probability of outputting the incorrect graph, $\mathcal{A}$ is the event that a graph satisfies condition one of the $\varepsilon$-typical hypergraph set and $\mathcal{G}_{\mathcal{A}}$ is the set of graphs such that $(1-\varepsilon) \bar{m} \leq m \leq \bar{m}(1+\varepsilon)$.
From the typicality conditions we have:
\begin{itemize}
    \item $\mathbb{P}[\mathcal{A}]=1-o(1)$
    \item $\log \left|\mathcal{G}_{\mathcal{A}}\right|=\binom{n}{k}H_2(q)(1+o(1))$
    \item $H(G \mid \mathcal{A}=$ true $)=\binom{n}{k} H_2(q)(1+o(1))$, where $H_2(q)=q \log \frac{1}{q}+(1-q) \log \frac{1}{1-q}$ is the binary entropy function.
\end{itemize}

We also have from \cite{li2019learning}:
\begin{itemize}
    \item $I(G ; \widehat{G} \mid \mathcal{A}=$ true $) \leq$ $I(G ; \mathbf{Y} \mid \mathcal{A}=$ true $)\le t \log 2$, where $t$ represents the total amount of queries. 
\end{itemize}
Together, yielding:
$$
P_{\mathrm{e}} \geq\left(1-\frac{t \log 2}{\binom{n}{k} H_2(q)}\right)(1+o(1)) .
$$

In our setting, $q \rightarrow 0$, thus $H_2(q)=\left(q \log \frac{1}{q}\right)(1+o(1))$, and hence
$$
P_{\mathrm{e}} \geq\left(1-\frac{t \log 2}{\frac{1}{k} q n^k \log \frac{1}{q}}\right)(1+o(1)) .
$$

Since $\bar{m}=\frac{1}{k} q n^k(1+o(1))$, we conclude that to have vanishing error probability we must have at least $\left(\bar{m} \log _2 \frac{1}{q}\right)(1-\eta)$ queries, for arbitrarily small $\eta>0$.

\end{proof} 
\newpage

\section{PROOFS FOR SECTION~\ref{sec:grotesque}}\label{sec:missing-proofs-grotesque}
In this section, we prove all the lemmas from Section~\ref{sec:grotesque} that are not proved in the body of the paper. We will make use of the following version of Hoeffding's inequality. 
\begin{lemma}\label{lem:hoeffding}
Let $X\sim$ Binom$(n,p)$. Then, for any $t >0$:
    \[
        \Pr\left[\left|{X\over n} - p 
        \right|\geq t\right] \leq 2e^{-2nt^2}.
    \]
\end{lemma}

\detectionProbability*
\begin{proof}%
    By assumption, the set $\mathcal{B}$ contains at least two distinct hyperedges. Fix two distinct hyperedges $h$ and $h'$. Then let:
    \begin{itemize}
        \item $D$ be the event that the set $S$ contains a full hyperedge $\mathcal{B}$,
        \item $D_h$ be the event that $S$ contains $h$, i.e. $h \subseteq S$
        \item $D_{h'}$ be the event that $S$ contains $h'$, i.e. $h' \subseteq S$,
        \item $D_{h\cap h'}$ be the event that $S$ satisfies $(h\cap h') \subseteq S$,
        \item $a$ be the cardinality of the intersection of $h$ and $h'$, i.e. $a:= |h\cap h'|$. Note that $a$ is some integer between $0$ and $k-1$.
    \end{itemize}
    Since each element of $\mathcal{B}$ is included in $S$ independently, the events $D_h$ and $D_{h'}$ are conditionally independent given $D_{h\cap h'}$. We then have:
    \begin{align*}
        \Pr[D] &\geq \Pr[D_h \cup D_{h'}] = \Pr[D_h \cup D_{h'} \mid D_{h\cap h'}] \Pr[D_{h\cap h'}]\\
        &=\left(1-\Pr[\overline{D_h} \cap \overline{D_{h'}}\mid D_{h\cap h'}]\right) \Pr[D_{h\cap h'}]\\
        &= \left(1-\Pr[\overline{D_h}\mid D_{h\cap h'}]\cdot \Pr[\overline{D_{h'}}\mid D_{h\cap h'}]\right) \left({1\over \sqrt[k]{e}}\right)^a\\
        &= \left[1-\left(1-\left({1\over \sqrt[k]{e}}\right)^{k-a}\right)^2\right] \left({1\over \sqrt[k]{e}}\right)^a\\
        &= {2 \over e} - \left({1\over e}\right)^{2k-a\over k} \geq {2 \over e} - \left({1\over e}\right)^{k+1\over k},
    \end{align*}
    as needed.
\end{proof}

\MultiplicityTestIsCorrect*
\begin{proof}%
    If the bundle contains no hyperedge, then the fraction of positive tests will be zero and the algorithm will return $0$ every time. Otherwise, Let $R_0, R_1$ be the events that the multiplicity test returns $0$ and $1$ respectively. If the bundle contains a single hyperedge, the probability that any given edge-detection test returns $1$ is $p_{single} = (e^{-k})^k = e^{-1}$. Applying Lemma~\ref{lem:hoeffding}, we obtain:
    \[
        \Pr[ R_1 ] \geq  \Pr\left[ \left|\hat{p} - {1\over e}\right| < {M\over 2}\right] \geq 1-2e^{-t_{mul}M^2/2} = 1-\delta.
    \]
    
    On the other hand, if the bundle contains at least two hyperedges, by Lemma~\ref{lem:detection-probability} the probability $p_{multiple}$ that any individual test detects a hyperedge satisfies:
    \begin{equation}\label{eq:probability-test-multiple-edge}
        p_{multiple} \geq {2\over e} - {1\over e^{(k+1)/k}}.
    \end{equation}
    Hence, in this case:
    \[
        \Pr[R_0] \geq \Pr\left[\left|\hat{p} - p_{multiple}\right| < M \right] \leq e^{-t_{mul} M^2/2} = 1-\delta.
    \]
    as needed.
\end{proof}

\MultiplicityTestIsEfficient*
\begin{proof}%
    The number of queries made is:
    \[
        t_{mul} = 2 \log {2 \over \delta} e^2 {e^{2/k} \over \sqrt[k]{e} -1} \leq e^3 k \log {2\over \delta},
    \]
    The decoding time is proportional to the number of queries.
\end{proof}

\section{ALGORITHMIC UPPER BOUNDS}\label{sec:missing-proofs-group-testing}
In this section we complete the analysis of the results in Section~\ref{sec:lower-query-complexity-algorithms}. We will make use of the assumption that $t$ represents the total number of tests, also referred to as queries. 

\TheoremCOMP*

\begin{proof}
    We will adapt the proof of \cite{li2019learning} of the analogous theorem for graphs. We note that conditioning on the random graph being in the typical set of graphs with high probability, we need only show that using the above stated amount of tests yields error probability approaching zero. So we examine the probability of failing to identify a non-edge, say $(i_1,\dots,i_k)\not\in E$, in the hypergraph setting this probability changes slightly. Recall, that there are two ways a test could fail to identify a non-edge: \begin{enumerate}
        \item at least one vertex in $(i_1,\dots,i_k)$ isn't present in the test
        \item $(i_1,\dots,i_k)$ is contained in the test, but another edge of $G$, our hypergraph, is also present in the test.
    \end{enumerate}
    This results in the probability of failing to identify $(i_1,\dots,i_k)$ as a non-edge as $$p_{ne}=(1-p^k)+p^kP_G[\{i_1, \dots, i_k\} \subseteq \mathcal{L}],$$
    where we recall that $\mathcal{L}$ is the set of nodes in the test and $p$ is the probability of inclusion in the test, and $P_G[\{i_1, \dots, i_k\} \subseteq \mathcal{L}]$ is the probability of a positive Bernoulli test, given $\{i_1, \dots, i_k\}$ is included in the test. If we are to have a positive test, we have either \begin{enumerate}
        \item An edge $e=(e_1,\dots e_k)\in E$ such that $e\cap (i_1,\dots,i_k) \ne \emptyset $,
        \item An edge $e=(e_1,\dots e_k)\in E$ such that $e\cap (i_1,\dots,i_k) = \emptyset $,
    \end{enumerate} 
    included in the test. In case 1, we can examine the situation $|e\cap (i_1,\dots,i_k)| =1$, noting that if we examine all possible neighbors, the rest of $(i_1,\dots,i_k)$ is included in that set. $$
\begin{aligned}
P_G[\{i_1, \dots, i_k\} \subseteq \mathcal{L}]&\leq \mathbb{P}\left[e\cap (i_1,\dots,i_k) =1 \mid\{i_1, \dots, i_k\} \subseteq \mathcal{L}\right]+\mathbb{P}\left[e\cap (i_1,\dots,i_k) = \emptyset \mid\{i_1, \dots, i_k\} \subseteq \mathcal{L}\right] \\
&\leq p^{k-1}k\Delta(G) + P_G\\
\end{aligned}
$$

In the first term, there are at most $k$ choices of vertices to intersect with, then $\Delta(G)$ possible edges containing that vertex. Since the second event is independent of conditioning on   $\{i_1, \dots, i_k\} \subseteq \mathcal{L}$, we get that the probability is just $P_G$. We assumed our graph is in the typical set so we can substitute in $P_G=\left(1-e^\nu\right)(1+o(1))$, we also have that $\Delta(G)kp^{k-1}=o(1)$, so our probabilities in the hypergraph case align with those in \cite{li2019learning}. We note that in the hypergraph case our bound changes slightly over union-bounding over a possible $\binom{n}{k}$ non-edges, resulting in $$
\mathbb{P}[\text { error }] \leq n^k e^{-\frac{t}{e m}(1+o(1))} .
$$
After re-arranging we have that $\mathbb{P}[\text { error }] \rightarrow 0$ as long as $$t\ge kem\log n(1+\eta)$$

for arbitrarily small $\eta>0$. Since $m=\bar{m}(1+o(1))$ for all typical graphs, and the probability that $G$ is typical tends to one, we obtain the condition in our statement.
\end{proof}

\TheoremDD*
\begin{proof}
    We again adapt the proof in \cite{li2019learning} of an analogous theorem for graphs. Once again, we have a hypergraph in the typical set, so we need only show that with the stated number of queries our error probability goes to zero. 

    There are two steps to the \texttt{DD} algorithm, the first in which we find a set of 'potential' edges, this set could include non-edges. Then the second where we find the set of edges from our set of potential edges. The argument examines how large $t$, the number of queries, must be for each of these events to occur with high probability. We adapt slightly the two events in the first step of the proof in \cite{li2019learning}.
    \begin{itemize}
        \item $H_0$ is the total number of non-edges in the potential edge set, PE
        \item $H_1$ is the number of non-edges in PE such that at least one of its vertices form a part of at least one true edge.
    \end{itemize}
    We have that the amount of non-edges must be less than $n^k$, and then utilizing the probability of failing to identify a non-edge recorded in the \texttt{COMP} algorithm proof above, we have $$
\mathbb{E}\left[H_0\right] \leq n^k e^{-\frac{t}{e m}(1+o(1))} .
$$
The amount of non-edges sharing a vertex with an edge is upper-bounded by $mkn^{k-1}$, but still must be less than $n^k$, so we have $$
\mathbb{E}\left[H_1\right] \leq \min \left\{m k n^{k-1}, n^k\right\} e^{-\frac{t}{e m}(1+o(1))} .
$$
Applying Markov's inequality, we have for any  $\xi_0>0$ and $\xi_1>0$ that 
\begin{align}
& \mathbb{P}\left[H_0 \geq n^{k \xi_0}\right] \leq n^{k\left(1-\xi_0\right)} e^{-\frac{t}{e m}(1+o(1))} \\
& \mathbb{P}\left[H_1 \geq n^{k \xi_1}\right] \leq \min \left\{m k n^{k-1}, n^k\right\} n^{-k \xi_1} e^{-\frac{t}{e m}(1+o(1))} \label{(8)}.
\end{align}

After re-arranging we find that these two probabilities go to zero as $n \rightarrow \infty$ as long as
$$
\begin{aligned}
& t \geq\left(k\left(1-\xi_0\right) e m \log n\right)(1+\eta), \\
& t \geq(1+\eta) e m \log n \times \begin{cases}k\left(1-\xi_1\right) & \frac{1}{k} \leq \theta<1 \\
k\left(1+\theta-\xi_1\right)-1 & 0<\theta \leq \frac{1}{k}\end{cases} \\
&
\end{aligned}
$$
for arbitrarily small $\eta>0$. The first case uses the $n^k$ term in the $\min\{\cdot\}$ term in \ref{(8)}. The second case uses the $mkd$ term in the $\min\{\cdot\}$ term and that $m=\Theta(n^{k\theta})$ and when $\theta> \frac{1}{k}$, $d_{max}\le kn^{k-1}q=kn^{k\theta-1}$, the last case we look at when $theta\le \frac{1}{k}$, so $d_{max}= O(\log n)$. 

We show that $H_0=o(m)$ and $H_1=o(\sqrt{m})$ (with high probability). By setting $\xi_0$ to be arbitrarily close to (but still less than) $\theta$, and similarly $\xi_1$ arbitrarily close to $\theta / 2$, the above requirements simplify to $$
\begin{gathered}
t \geq(k(1-\theta) e m \log n)(1+\eta), \\
t \geq(1+\eta) e m \log n \times \begin{cases}k-k\theta/2 & \frac{1}{k}\leq \theta<1 \\
k+k\theta/2-1 & 0<\theta \leq \frac{1}{k}\end{cases}
\end{gathered}
$$
for arbitrarily small $\eta>0$.

    The second step of the algorithm can be shown to succeed as long as $$
t \geq(k \theta e m \log n)(1+\eta)
$$
for arbitrarily small $\eta>0$. This follows directly from the proof of the second step written by \cite{li2019learning}. 

\end{proof}

\TheoremSSS* 
\begin{proof}

The proof of the \texttt{SSS} algorithm bound in the hypergraph case follows extremely closely to the graph case, where the main difference is simply replacing the number of vertices with general arity term $k$.  We will just verify assumptions hold that are slightly altered in our setup. In the hypergraph case, event $A_1$ is just the event that another hyperedge intersects with the hyperedge we're seeking to find, say $(i_1, \dots i_k)$, and masking it. Therefore, $$\mathbb{P}\left[A_1^{(i_1, \dots i_k)} \mid\{i_1, \dots i_k\} \subseteq \mathcal{L}\right] \leq \Delta(G) (1+ p)^{k}$$  and defines $\xi^{\prime}=\Delta(G)  (1+ p)^{k}$. We recall that the converse bound we are trying to prove is $t=\Omega(m \log n)$, so we can assume without loss of generality that $t=\Theta(m \log n)$, as additional tests only improve the \texttt{SSS} algorithm. From \cite{li2019learning} we can assume without loss of generality that $p^k=\Theta\left(\frac{1}{m}\right)$, since if $p^k$ behaves as $o\left(\frac{1}{m}\right)$ or $\omega\left(\frac{1}{m}\right)$ then the probability of a positive test tends to 0 or 1 as $n \rightarrow \infty$, and it follows from a standard entropy-based argument that $\omega(m \log n)$ tests are needed. We claim that these conditions imply that
$$
\frac{e^{-2 t\left(p^k \xi+O\left(p^{2k}\right)\right)}}{e^{2 t p^k \xi^{\prime}}} \rightarrow 1
$$ 

We note that $\xi^{\prime}$ still behaves as $O\left(n^{-c}\right)$ for sufficiently small $c$. We also note that we have the same behavior for $\xi=(1+k\Delta(G) )p^k$. This is seen by noting that $t p^k=\Theta(\log n)$ by the above-mentioned behavior of $t$ and $p^k$.  This behavior fall inline with the $O\left(t p^{2k}\right)$ term by the above-mentioned behavior of $t$ and $p^k$, and is seen to also hold for $\xi$ and $\xi^{\prime}$ by noting that $\Delta(G)  p^{k-1}=\Theta\left(\frac{\Delta(G) \sqrt[k]{m}}{m}\right)$, along with $\Delta(G) =O(\max \{\log n, n^{k-1} q\})$, $m=\Theta(n^k q)$, and the behavior of $q$. Thus, nothing of consequence changes when generalizing to hypergraphs. 

\end{proof}
\end{document}